\documentclass[twoside,leqno]{article}

\usepackage[letterpaper]{geometry}

\usepackage{siamproceedings}

\usepackage[T1]{fontenc}
\usepackage{amsfonts}
\usepackage{amssymb}
\usepackage{graphicx}
\usepackage{epstopdf}
\usepackage{enumitem}
\usepackage{algorithmic}
  \DeclareGraphicsExtensions{.eps,.pdf,.png,.jpg}

\newsiamremark{remark}{Remark}
\newsiamremark{problem}{Problem}
\newsiamthm{claim}{Claim}
\newsiamthm{example}{Example}

\usepackage{amsopn}

\begin{document}

\newcommand\relatedversion{}
\renewcommand\relatedversion{\thanks{The full version of the paper can be accessed at \protect\url{https://arxiv.org/abs/2504.13268}}} 

\title{\Large Dichotomy for orderings?\thanks{The full version of the paper can be accessed at \protect\url{https://arxiv.org/abs/2504.13268}}}
    \author{G\'abor Kun \\ kungabor@renyi.hu
    \thanks{HUN-REN Alfr\'ed R\'enyi Institute of Mathematics and
    E\"otv\"os L\'or\'and University \\
    The first author's work has been supported by the Hungarian Academy of Sciences Momentum Grant No. 2022-58 and ERC Advanced Grant ERMiD.}
    \and Jaroslav Ne\v{s}et\v{r}il \\ nesetril@iuuk.mff.cuni.cz \thanks{Department of Applied Mathematics (KAM)
    and Computer Science Institute (IUUK),
    Charles University \\
    The second author's work has been supported by DIMATIA of Charles University Prague and by ERC under grant DYNASNET, grant agreement No. 810115
}}
\date{}

\maketitle


\fancyfoot[R]{\scriptsize{Copyright \textcopyright\ 2026 by SIAM\\
Unauthorized reproduction of this article is prohibited}}


\fancyfoot[R]{\scriptsize{Copyright \textcopyright\ 2026\\
Copyright for this paper is retained by authors}}



\begin{abstract} 
Fagin defined the class $NP$ by the means of Existential Second-Order logic. Feder and Vardi expressed it (up to polynomial equivalence) by special fragments of Existential Second-Order logic (SNP), while the authors used forbidden expanded substructures (cf. lifts and shadows). Consequently, for such problems there is no dichotomy, unlike for CSPs.

We prove that {\it ordering problems} for graphs defined by finitely many forbidden ordered subgraphs capture the {\it full power} of the class $NP$, that is, any language in the class $NP$ is polynomially equivalent to an ordering problem. In particular, we refute  a conjecture of Hell, Mohar and Rafiey that dichotomy holds for this class. On the positive side, we confirm the conjecture of Duffus, Ginn and R\" odl that ordering problems defined by a single obstruction which is a biconnected ordered graph are $NP$-complete if the graph is not complete. 

We initiate the study of meta-theorems for classes which have the {\it full power} of the class $NP$.
For example, homomorphism problems (or CSPs) do not have full power (similarly to coloring problems). On the other hand, we show that problems defined by the existence of an ordering, which avoids certain ordered patterns, have full power.
We find it surprising that such simple structures can express the full power of $NP$.

It is essential that we treat these problems in a more general context. An interesting feature appeared: while the full power is reached by disconnected structures, and one can even guarantee the connectivity of all patterns, this is no longer the case for biconnected patterns. We prove that we have here a general phenomenon: For finite sets of biconnected patterns (which may be colored structures or ordered structures) dichotomy holds, while for general patterns we have full power. A principal tool for obtaining these results is the {\it Sparse Incomparability Lemma} in many of its variants, which are classical results in the theory of homomorphisms of graphs and structures. We prove it here in the setting of ordered stuctures as a Temporal Sparse Incomparability Lemma. This is a non-trivial result, even in the random setting, and a deterministic algorithm requires more effort. Interestingly, our proof involves the Lov\'asz Local Lemma.

Dichotomy results for forbidden biconnected patterns encourage to prove that the ordering problem for any non-trivial biconnected graph is $NP$-complete (as conjectured by Duffus, Ginn and R\" odl). We confirm this by bringing together most of the techniques developed in the paper, and we also use the results of Bodirsky and K\'ara on the complexity of temporal CSPs. 
\end{abstract}

\section{Introduction and main results.} We assume $P \neq NP$ throughout this paper.
The study of the class $NP$ usually focuses on the two extremes, tractable problems and $NP$-complete problems. Ladner \cite{Lad} showed that there are intermediate problems in $NP$ which are neither tractable nor $NP$-complete. Feder and Vardi \cite{FV} investigated subclasses of $NP$
in terms of second-order logic 
searching for a large class that may admit dichotomy. The class of Constraint Satisfaction Problems (CSP) became their natural candidate. 
Bulatov \cite{Bul} and Zhuk \cite{Zh} proved the Feder-Vardi dichotomy conjecture. 

Here we propose the study of those classes, which may admit dichotomy or have the full computational power of the class $NP$ with the aim  to clarify the boundary between classes admitting dichotomy and those classes which have the full power of $NP$. It is particularly interesting  to study these problems in the combinatorial context of orderings and restricted colorings of graphs. Surprisingly, connectivity plays a major role here.

An {\it ordered graph}, denoted by ${\bf G}^<$, is an undirected graph ${\bf G}$ with a fixed linear ordering $<$ of its vertices. Similarly, we denote by ${\mathcal{F}}^<$ a set of ordered graphs. For a fixed set of ordered graphs ${\mathcal{F}}^<$ we consider the following decision problems:
\vspace{5mm}

\noindent
{\bf ${\mathcal{F}}^<$-free ordering problem}

\noindent
Given a graph $\bf G$ does there exist an ordering $<$ of the vertices of $\bf G$ such that
$\bf{G}^<$ does not contain an ordered subgraph  isomorphic to $\bf{F}^<$ for any $\bf{F}^< \in {\mathcal{F}}^<$?
\vspace{5mm}

\noindent
{\bf Induced ${\mathcal{F}}^<$-free ordering problem}

\noindent
Given a graph $\bf G$ does there exist an ordering $<$ of the vertices of $\bf G$ such that
$\bf{G}^<$ does not contain an ordered induced subgraph isomorphic to $\bf{F}^<$ for any $\bf{F}^< \in {\mathcal{F}}^<$?

\vspace{5mm}

Observe that induced ordering problems have greater expressive power than ordering problems. We formulate this in the following way:

\begin{remark} \label{induced}
For every finite set of ordered graphs $\mathcal{F}^<$ there exists a finite set of ordered graphs $\mathcal{G}^<$ such that the $\mathcal{F}^<$-free ordering problem and the induced  $\mathcal{G}^<$-free ordering problem are the same languages. 
For example, $\mathcal{G}^<$ can be taken the set of supergraphs of graphs in $\mathcal{F}^<$, i.e., $\mathcal{G}^<=\{{\bf G}^<: \exists {\bf F}^< \in \mathcal{F}^<, V(G)=V(F), <_G=<_F, E({\bf G}) \supseteq E({\bf F}) \}$).
\end{remark}

The interplay between ordered and unordered structures is interesting from the structural as well as the algorithmic point of view. From the structural side one can mention the relationship to posets and their diagrams \cite{B,NR2}, for the relationship to Ramsey theory (``order property'')  see \cite{N,B}, while for the statistics of orderings see \cite{MaT,NR1,NR2,BBJ} with applications to unique ergodicity \cite{AKL}. 

From the computational point of view one can mention results relating chromatic numbers to orderings starting with the classical results of Gallai, Hasse, Minty, Roy and Vitaver (see, e.g., \cite{HN} but also \cite{M}). 
This was considered in the algorithmic context by Duffus, Ginn and R\" odl \cite{DGR} and by Hell, Mohar and Rafiey \cite{HM}, where various complexity results were obtained and some conjectures were formulated, see also \cite{D,FH,GHH}. Note that such problems may be $NP$-complete even for very simple ordered graphs. For example, for the monotone path of length $k$ the ordering problem is equivalent to having chromatic number at most $k$, and hence $NP$-complete. 

Hell, Mohar and Rafiey \cite{HM} have conjectured that induced ordering problems always have dichotomy and proved it in several cases. Our first main result refutes their conjecture.

\begin{theorem}
\label{main1}
For every language $L$ in the class $NP$ there exists a finite set $\mathcal{F}^<$ such that the $\mathcal{F}^<$-free ordering problem and $L$ are polynomially equivalent.
\end{theorem}

Shortly, finitely many forbidden ordered graphs determine (up to polynomial equivalence) any language in $NP$. In other words, the class of $\mathcal F^<$-free ordering problems has the {\it full computational power} of the class $NP$. By Remark \ref{induced} this holds in the induced case too:

\begin{corollary} \label{main1induced}
For every language $L$ in the class $NP$ there exists a finite set $\mathcal{F}^<$ such that the induced $\mathcal{F}^<$-free ordering problem and $L$ are polynomially equivalent.
\end{corollary}

Thus, using Ladner's theorem \cite{Lad} we can refute the conjecture of \cite{HM}. 

\begin{corollary} \label{main11}
There is no dichotomy for the induced $\mathcal F^<$-free ordering problems.
\end{corollary}

One can define coloring problems similarly to ordering problems. A {\it colored graph}, denoted by ${\bf G}'$, is an undirected graph ${\bf G}$ with a fixed coloring of its vertices. We denote by ${\mathcal{F}}'$ a set of colored graphs. For a fixed set of colors and colored graphs ${\mathcal{F}}'$ we consider the following decision problems:
\vspace{5mm}

\noindent
{\bf ${\mathcal{F}}'$-free coloring problem}

\noindent
Given a graph $\bf G$ does there exist a coloring of the vertices of $\bf G$ such that the colored graph $\bf{G}'$ does not contain a colored subgraph  isomorphic to $\bf{F}'$ for any $\bf{F}' \in {\mathcal{F}}'$?
\vspace{5mm}

For such coloring problems we have the analogue of Theorem \ref{main1}, formulated as Theorem \ref{graphs}. In particular, Theorem \ref{graphs} also answers a question of Guzm\'an-Pro (Question 6.3., \cite{G}).
Theorem \ref{graphs} is an important milestone towards proving Theorem \ref{main1}.



One can also prove that there is no dichotomy for {\it connected} ordered graphs, see Theorem \ref{connected}, and for {\it connected} colored graphs, see Theorem \ref{colcon}.

This is interesting since the landscape is fundamentally different in the biconnected case as our second main result shows. The definition of the $\mathcal{F}^<$-free ordering problem extends to relational structures in a straightforward way. A relational structure is biconnected if its Gaifman graph is biconnected ($2$-connected).

\begin{theorem} \label{main2}
Let $\mathcal F^<$ be a finite set of finite biconnected relational structures of the same type equipped with an ordering. Then the $\mathcal F^<$-free ordering  problem is either $NP$-complete or tractable, and the induced $\mathcal F^<$-free ordering  problem is also either $NP$-complete or tractable.
\end{theorem}

Theorem \ref{main2} also holds in terms of colorings of relational structures, see the full version for the precise statement. As finite CSPs are defined by colorings using single relational tuples, and these are biconnected, this generalizes the dichotomy theorem of Bulatov \cite{Bul} and Zhuk \cite{Zh}.  

Most ordering problems seem to be $NP$-complete. We will add one exact result:
Duffus, Ginn and R\" odl \cite{DGR} conjectured that if $\mathcal{F}^<$ consists of a single ordered biconnected graph that is not complete then the induced $\mathcal{F}^<$-free ordering problem is $NP$-complete. We give a characterization of tractable ordering problems defined by a single biconnected graph and verify their conjecture.

\begin{theorem}
\label{main3}
Let ${\bf F}^<$ be a finite biconnected ordered graph that is not complete. Then the $\{ {\bf F}^< \}$-free ordering problem and the induced $\{ {\bf F}^< \}$-free ordering problem are both $NP$-complete. 
\end{theorem}

The main tool in the proof of Theorem \ref{main2} is the {\it Temporal Sparse Incomparability Lemma} (Theorem \ref{temporalSIL}), the main technical result of the paper.
The connection of the SIL and biconnectivity goes back to \cite{FV}. 
Using SIL, they proved a randomized reduction of a finite CSP to the CSP restricted to structures with large girth. This was derandomized by the first author \cite{K}.
We exploit this idea in a much more general context of forbidden patterns defined either by orderings or by (potentially infinite) colorings.
Our paper highlights the role of SIL for temporal CSPs, which is proved in the setting of orderings of relational structures by a novel application of the Lov\'asz Local Lemma. Several natural problems are motivated by our paper. Let us mention here just one:

\begin{problem}
Do families defined by forbidden ordered trees (forests) admit a dichotomy? Or perhaps they have full power?
\end{problem}

The paper is organized as follows.  We introduce the necessary definitions in Section \ref{secnotation}. In Section \ref{secexamples} we give two typical examples of our results for biconnected patterns and, as a warm-up, we sketch the proofs for these.
In Section \ref{secfull1} we prove Theorem \ref{main1} on the full power of orderings and its colored version, Theorem \ref{graphs}: these show the relationship between coloring and ordering problems, and give a feeling when might these be extended to prove full power of further classes.
In Section \ref{secTSIL} we state the Temporal SIL and give a probabilistic construction.

\section{Notions and Notation}\label{secnotation}

For a relational symbol $R$ and relational structure $\bf A$ let $A$ denote the universe of $\bf A$ and let $R(\bf A)$ denote
the set of tuples of $\bf A$ which belong to $R$. Similarly, as for ordered graphs, we define the ordered relational structure $\bf A^<$ as a structure $\bf A$ with an ordering $<$, sometimes denoted by $<_{\bf A}$. A finite interval w.r.t. the ordering is a set of consecutive elements.
The relational structure $\bf A$ is called {\it temporal} if  $A=\mathbb{Q}$ and every $r$-ary relation $R \subset \mathbb{Q}^r$ is invariant under every automorphism of $(\mathbb{Q},<)$, i.e., the quasiorder of the elements in an $r$-tuple tells if the tuple is in $R$.

Let $\tau$ denote the {\it signature} (type) of relational symbols,
and let $Rel(\tau)$ denote the class of all finite relational structures
with signature $\tau$. We will often work with two (fixed)
signatures, $\tau$ and $\tau \cup \tau'$, where the signatures
$\tau$ and $\tau'$ are always supposed to be disjoint and $\tau'$ consists of monadic relational symbols. For
convenience, we denote structures in $Rel(\tau)$ by $\bf A,\bf B$
etc. and structures in $Rel(\tau \cup \tau')$ by $\bf A',\bf B'$
etc. We will denote $Rel(\tau \cup \tau')$ by $Rel(\tau,
\tau')$. The classes $Rel(\tau)$ and $Rel(\tau, \tau')$ will
be considered as categories endowed with all homomorphisms. Recall
that a homomorphism is a mapping which preserves  all relations.
Just to be explicit, for relational structures ${\bf A,\bf B} \in
Rel(\tau)$ a mapping $f: A \longrightarrow B$ is a
{\it homomorphism} ${\bf A} \longrightarrow \bf B$ if for every
relational symbol $R \in \tau$ and for every tuple $(x_1, \ldots,
x_t) \in R(\bf A)$ we have $(f(x_1), \ldots, f(x_t)) \in R(\bf B)$.
Similarly, we define homomorphisms for the class $Rel(\tau, \tau')$.

The
following special subclass of $Rel(\tau, \tau')$ will be
important: denote by $Rel^{cov}(\tau, \tau')$ the class of all
structures in $Rel(\tau, \tau')$ where we still assume that all
relations in $\tau'$ have arity one, and that every element of a structure is contained by some relation in $\tau'$. 
The class  $Rel^{cov}(\tau, \tau')$ corresponds to structures in $Rel^{cov}(\tau)$ together with some coloring of all its elements. 

We will also work with other similar 
categories. We denote by $Rel_{inj}(\tau)$ the category where the objects are again the
finite relational structures of type $\tau$, but the morphisms are the
injective homomorphisms ${\bf A} \hookrightarrow {\bf B}$. We denote by $Rel^{cov}_{inj}(\tau,\tau')$ the 
corresponding category containing the
same class of objects as $Rel^{cov}(\tau,\tau')$. 

Let $\mathcal F$ be a set of structures in one of the above categories. We denote by ${\rm
Forb}(\mathcal F)$ the class of all structures ${\bf A}$ satisfying ${\bf F} \not\longrightarrow \bf
A$ for every ${\bf F} \in \mathcal F$.
Combining the above notions we can consider the class $\Phi(Forb_{inj}(\mathcal{F'}))$ which is the class of all objects $\bf A$ for which there exists an $\bf A'$ that does not contain any $\bf F' \in {\mathcal F'}$ as a substructure. Classes defined in this way are central to this paper. We often refer to $\Phi(Forb_{inj}(\mathcal{F'}))$ as the language of the $\mathcal{F'}$-free coloring problem viewing $\tau'$ as the set of colors.

Combining the above notions we can consider the class $\Phi(Forb_{inj}(\mathcal{F'}))$ which is the class of all objects $\bf A$ for which there exists a lift $\bf A'$ which does not contain any $\bf F' \in {\mathcal F'}$ as a substructure. Classes defined in this way are central to this paper. We often refer to $\Phi(Forb_{inj}(\mathcal{F'}))$ as the language of the $\mathcal{F'}$-free coloring problem viewing $\tau'$ as the set of colors.

The following basic lemma might be folklore.

\begin{lemma} \label{bounded}
For every temporal relational structure ${\bf T}$ of finite type there exists $D$ such that $CSP({\bf T})$ can be polynomially reduced to its restriction to relational structures with maximum degree at most $D$.
\end{lemma}

\section{The biconnected phenomenon by two examples} \label{secexamples}

We give two examples to illustrate the proof that a finite set of finite biconnected patterns (subgraphs with a given coloring or ordering) leads to a CSP and hence to dichotomy. 
First, we consider a coloring problem: 

\begin{example}
Consider the language $L$ of undirected graphs admitting a two-coloring of the vertices without a monochromatic triangle. 
What is the complexity of $L$?
\end{example}

Thus, we have two colors and the forbidden patterns are the two monochromatic triangles.
Consider $NAE$ (Not-All-Equal SAT) or, equivalently, the uniform $3$-hypergraph $2$-coloring problem. Clearly $L$ can be reduced to $NAE$, by assigning to a graph the $3$-hypergraph on its vertex set, where we impose a $3$-hyperedge on every triangle of the graph: the good $2$-colorings of this hypergraph are exactly the good colorings of the graph, i.e., colorings avoiding monochromatic triangles. 

On the other hand, given a $3$-hypergraph ${\bf H}$, we can assign to it a graph ${\bf G}$ on the same base set by replacing every hyperedge by a triangle. Unfortunately, this might not be a reduction of $NAE$ to $L$, since ${\bf G}$ can have triangles that do not originate from a single hyperedge. However, if the girth of ${\bf H}$ is at least four then this can not happen: every triangle of ${\bf G}$ is contained by a $3$-hyperedge of {\bf H}. Thus, the good colorings of ${\bf H}$ are exactly the good colorings of ${\bf G}$. We know from \cite{K} that $NAE$ is polynomially equivalent to the restriction of $NAE$ to relational structures with girth at least four. Thus, $NAE$ and $L$ are polynomially equivalent, and $L$ is $NP$-complete.

In the second example we consider an ordering problem corresponding to a single ordered graph on four vertices. Our proof is similar, but it involves many new elements: an interplay of orderings with forbidden colored subgraphs, using the rational numbers as colors and so leading to temporal CSPs, and the Temporal SIL. The following example corresponds to a particular case of Theorem \ref{main2}.

\begin{example}
Consider the ordered undirected graph 
${\bf F}^<$ on $F= \{ 1,2,3,4 \}$, where the ordering is the natural ordering and every distinct pair is in relation but $(1,3)$ and $(3,1)$ (i.e., ${\bf F}^<$ is the undirected complete graph on $\{ 1,2,3,4 \}$ without the edge $(1,3)$). What is the complexity of the $\{ {\bf F}^< \}$-free ordering problem?
\end{example}

For a finite set $S$ we say that two injective mappings $\varphi_1, \varphi_2: S \hookrightarrow \mathbb{Q}$ are equivalent if $\varphi_1(x)<_{\mathbb{Q}}\varphi_1(y) \iff \varphi_2(x)<_{\mathbb{Q}}\varphi_2(y)$ for every $x,y \in S$.
Then the orderings of a finite set $S$ are in one-to-one correspondence with equivalence classes of injective mappings to $\mathbb{Q}$. 

First, we reformulate the $\{{\bf F}^<\}$-free ordering problem as a coloring problem with forbidden colored subgraphs. Let $\mathbb{Q}$ be the set of colors and let $\mathcal{F}'$ contain every coloring ${\bf F}'$ of ${\bf F}$, where the four vertices get pairwise distinct colors and the order of these rational numbers defines an ordered graph isomorphic to ${\bf F}^<$, or where any pair of vertices gets the same color. 

Consider the temporal CSP with base set $T=\mathbb{Q}$ and with one quaternary relation $R({\bf T})$ such that $(q_1, q_2, q_3, q_4) \notin R({\bf T})$ if either $q_i = q_j$ for a pair $i \neq j$, or the ordering satisfies any of the following four chain on inequalities: either $q_1<q_2<q_3<q_4$ or $q_3<q_2<q_1<q_4$ or $q_1<q_4<q_3<q_2$ or $q_3<q_4<q_1<q_2$. Note that these orderings correspond to the automorphisms of the graph ${\bf F}$: when forbidding ${\bf F}^<$ with its standard ordering we also forbid these ordered graphs.

We reduce the $\{{\bf F}^<\}$-free ordering problem to $CSP({\bf T})$. We assign to a finite undirected graph ${\bf G}$ the structure ${\bf S}$ on $S=G$ with one single quaternary relation $R({\bf S})$, where $(x_1, x_2, x_3, x_4) \in R({\bf S})$ iff the mapping $i \mapsto x_i$ is an embedding ${\bf F} \hookrightarrow {\bf G}$. It is easy to see that injective mappings $G \hookrightarrow \mathbb{Q}$ inducing a good ordering are exactly the injective homomorphisms ${\bf S} \hookrightarrow {\bf T}$.

How about a non-injective homomorphism ${\bf S} \to {\bf T}$? The restriction of every homomorphism to a tuple in relation $R({\bf S})$ has to be injective, so for a non-injective homomorphism ${\bf S} \to {\bf T}$ a small perturbation gives an injective homomorphism. Therefore, the $\{{\bf F}^<\}$-free ordering problem has a polynomial time reduction to $CSP({\bf T})$.

On the other hand, given a finite relational structure ${\bf S}$ with one quaternary relation $R({\bf S})$ and girth greater than four assign the undirected graph ${\bf G}$ to it, where $G=S$ and we impose a copy of ${\bf F}$ on every tuple in $R({\bf S})$. Since ${\bf F}$ is biconnected and its size is less than the girth there are no other copies of ${\bf F}$ in ${\bf G}$, but those induced by the tuples in $R({\bf S})$. Hence (equivalence classes of) injective homomorphisms ${\bf S} \hookrightarrow {\bf T}$ correspond to good orderings for the $\{{\bf F}^<\}$-free ordering problem, and non-injective homomorphisms can be changed to injective homomorphisms by a small perturbation. We can conclude that $CSP({\bf T})$ restricted to structures with girth greater than four has a polynomial time reduction to the $\{{\bf F}^<\}$-free ordering problem. The reduction of $CSP({\bf T})$ to its restriction to relational structures with girth greater than four follows from the Temporal SIL, what will be explained in Section \ref{secTSIL}.

In order to show that the $\{{\bf F}^<\}$-free ordering problem is $NP$-complete, i.e., the Duffus-Ginn-R\"odl conjecture holds for it, we check that all the possible algebraic witnesses for tractability of temporal CSPs in \cite{BK} fail.

\section{Classes with full power}

\subsection{$\mathcal{F}'$-free coloring problems have full power} \label{secfull2}

The graph coloring problems have the full power of $NP$. The authors proved this for relational structures in \cite{KN}. We will use this to prove Theorem \ref{main1}.

\begin{theorem} \label{graphs}
For every language $L \in NP$ there exists a finite set of colors $C$
and a finite set of $C$-colored undirected graphs $\mathcal{F'}$ such 
that $L$ is polynomially equivalent to 
the $\mathcal{F'}$-free coloring problem.
\end{theorem}

Incidently, the validity of Theorem \ref{graphs} has been asked recently by Guzm\'an-Pro (Question 6.3., \cite{G}) in category theoretic terms. 

\begin{proof}(of Theorem \ref{graphs})
We know by \cite{KN} that there exist relational types $\tau, \tau'$ and a finite set of relational structures $\mathcal{S'} \subset Rel(\tau,\tau')$ such that $L$ is polynomially equivalent to $M=\Phi(Forb^{cov}_{inj}(\mathcal{S'}))$, see Section \ref{secnotation} for the notation.

We will construct a finite set $\mathcal{F'}$ of colored undirected graphs such that $\Phi(Forb^{cov}_{inj}(\mathcal{F'}))$ is polynomially equivalent to $M$. Let $R_1, \dots ,R_{|\tau|}$ denote the relational symbols in $\tau$ with arities $r_1, \dots ,r_{|\tau|}$, respectively.

Set $K=|\tau|+r+3$, where $r$ is the maximum arity of relational symbols in $\tau$.
We will consider the following undirected graph ${\bf G}_i$ for every relational symbol $R_i$. Let the vertex set of the graph ${\bf G}_i$ contain a cycle of length $K$ with vertices denoted by $v_1, \dots , v_K$,
where $v_j$ is adjacent to $v_{j+1}$ for every $j$, and $v_K$ is adjacent to $v_1$. Let $v_1$ be also adjacent to $v_{K-1}$ and $v_{K-2}$. And for every $i \leq j \leq i+r_i-1$ let $v_i$ be the starting vertex of a path with $K$ vertices in such a way that these paths are all vertex-disjoint and only share their starting vertex with the cycle. We will refer to the other endvertex of such a path not on the cycle as {\it root}. 

Now we define $\mathcal{F'}$. The set of colors will be $C=\tau'$.
Let $\mathcal{F'}$ consist of the following colored undirected graphs.

\begin{enumerate}
 
\item
{Every coloring of the graph ${\bf G}_i$ plus an additional edge (connecting two non-adjacent vertices).}

\item
{Every coloring of the graph ${\bf G}_i$ plus an additional vertex adjacent to one of the vertices of ${\bf G}_i$ that is not a root.}

\item
{For every ${\bf S'} \in \mathcal{S'}$ we define the following graph ${\bf G}_{\bf S}$. The vertex set $G_{\bf S}$ contains $S$, the base set of ${\bf S}$. And for every relational tuple  $(t_1, \dots , t_{r_i})$ in relation $R_i$ on ${\bf S}$ we add a copy of ${\bf G}_i$ such that these copies are vertex-disjoint apart from the roots, and the root that is the endvertex of the $j$th path is exactly $t_j$. We include ${\bf G'_S}$ in $\mathcal{F}'$ with every coloring that extends the coloring of ${\bf S'}$ on $S$.}
\end{enumerate}

Note that $(1)$ and $(2)$ forbid eventually subgraphs, i.e., every coloring of certain subgraphs.

Set $N=\Phi(Forb^{cov}_{inj}(\mathcal{F'}))$. We will show that $M$ and $N$ are polynomially equivalent.

First we reduce $M$ to $N$. Let ${\bf T}$ be a relational structure of type $\tau$. We construct an undirected graph ${\bf G}$ as follows. Let the vertex set of ${\bf G}$
contain the base set $T$ of ${\bf T}$ plus for every tuple in ${\bf T}$ of type $R_i$ a copy of ${\bf G}_i$ such that the roots are all in $T$, and else these copies are pairwise vertex-disjoint. If the relational tuple $(t_1, \dots , t_{r_i})$ is in relation $R_i$ on ${\bf T}$ then the roots of the corresponding copy of ${\bf G}_i$ in the base set are $t_1, \dots , t_{r_i}$. See Figure \ref{dancing}.

\begin{figure}[!h]\label{dancing}
\begin{center}
\includegraphics[width=10.0cm]{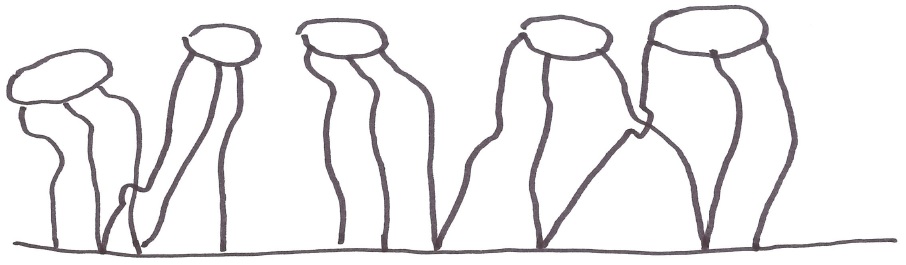}
\caption{The graph ${\bf G}$}
\end{center}
\end{figure}

\begin{claim}
${\bf T} \in M \iff {\bf G} \in N$. Moreover, 
given ${\bf T'}$, we have ${\bf T'} \in Forb^{cov}_{inj}(\mathcal{S'})$ if and only if there exists ${\bf G'} \in Forb^{cov}_{inj}(\mathcal{F'})$ such that the coloring of ${\bf G'}$ extends the coloring of ${\bf T'}$.
\end{claim}

\begin{proof}
Indeed, if ${\bf T'} \notin Forb^{cov}_{inj}(\mathcal{S'})$ then there is an ${\bf S'} \in \mathcal{S'}$ that admits an injective homomorphism $\iota: {\bf S'} \hookrightarrow {\bf T'}$. Thus, ${\bf G}_{\bf S}'$ and an injective homomorphism ${\bf G}'_{\bf S} \hookrightarrow {\bf G}'$ that agrees with $\iota$ on $S$ witness that ${\bf G'} \notin Forb^{cov}_{inj}(\mathcal{F'})$. 

Now assume that ${\bf G'} \notin Forb^{cov}_{inj}(\mathcal{F'})$. The construction of ${\bf G}$
guarantees that it admits no subgraphs of types $(1)$ and $(2)$ from $\mathcal{F'}$. Note that every cycle with length at most $K$ in ${\bf G}$ is contained by the homomorphic image of a ${\bf G}_i$, where the homomorphism is injective on the non-root vertices, since the paths from the cycle in ${\bf G}_i$ have length $K$. Hence every subgraph of ${\bf G}$ that is the homomorphic image of a graph ${\bf G}_i$, where the homomorphism is injective on the non-root vertices, corresponds to a tuple of ${\bf S}$ in relation $R_i$.
There is a colored graph ${\bf G'_S}$ for an ${\bf S} = \Phi({\bf S}'), {\bf S}' \in \mathcal{S'}$ and an injective homomorphism $\iota:{\bf G'_S}\hookrightarrow {\bf G'}$ witnessing that ${\bf G'} \notin Forb^{cov}_{inj}(\mathcal{F'})$. Hence $\iota|_S: {\bf S'} \hookrightarrow {\bf T'}$ witnesses that ${\bf T'} \notin Forb^{cov}_{inj}(\mathcal{S'})$.
\end{proof}

Now we prove that $N$ has a polynomial time reduction to $M$. Consider a graph ${\bf G}$, we may assume that ${\bf G}$ contains no copy of ${\bf G}_i$ plus one more edge from a non-root vertex (to an external or internal vertex), otherwise the graphs of type $(1)$ and $(2)$ witness that ${\bf G} \notin M$. Consider the set of vertices which are the image of a root in a graph ${\bf G}_i$ under an injective homomorphisms: the base set $T$ of ${\bf T}$ consists of these vertices. And for every copy of ${\bf G}_i$, where the roots are $t_1, \dots , t_{r_i} \in T$, add the tuple $(t_1, \dots , t_{r_i})$ to the relation $R_i$ on ${\bf T}$. 

\begin{claim} 
${\bf T} \in M \iff {\bf G} \in N$, moreover,
given the colored graph ${\bf G'}$ and ${\bf T'}$ obtained by the restriction of the coloring of ${\bf G}$ to $T$, the equivalence 
${\bf T'} \in Forb^{cov}_{inj}(\mathcal{S'}) \iff {\bf G'} \in Forb^{cov}_{inj}(\mathcal{F'})$ holds.
\end{claim}

\begin{proof}
If ${\bf G'} \notin Forb^{cov}_{inj}(\mathcal{F'})$ then there is a colored graph ${\bf G'}_{\bf S}$ and an injective homomorphism $\iota: {\bf G'}_{\bf S} \hookrightarrow {\bf G'}$ witnessing it. Now $\iota|_S: {\bf S'} \hookrightarrow {\bf T'}$ shows that ${\bf T'} \notin Forb^{cov}_{inj}(\mathcal{S'})$.

On the other hand, if ${\bf T'} \notin Forb^{cov}_{inj}(\mathcal{S'})$ then an ${\bf S'} \in \mathcal{S'}$ and an injective homomorphism $\iota:{\bf S'} \hookrightarrow {\bf T'}$ witness it. For every relational tuple of type $R_i$ in ${\bf S}$ there is a corresponding copy of the graph ${\bf G}_i$ whose roots are the coordinates of the tuple. Thus, there is an injective homomorphism $\kappa:{\bf G}_{\bf S} \hookrightarrow {\bf G}$ such that the inequality $\kappa|_S=\iota$ holds for the restriction to the roots. 
The injective mapping $\kappa:{\bf G}'_{\bf S} \hookrightarrow {\bf G}$ witnesses, for any extension of the coloring of ${\bf S'}$ to ${\bf G}'_{\bf S}$, that ${\bf G'} \notin Forb^{cov}_{inj}(\mathcal{F'})$.
\end{proof}

This completes the proof of the Theorem \ref{graphs}.
\end{proof}

The following connected version of Theorem \ref{graphs} also holds. Note that this is no longer true in the biconnected case as explained in the introduction.

\begin{theorem} \label{colcon}
For every language $L$ in the class $NP$ there exists a finite set of colors $C$ and a finite set $\mathcal{F}'$ of finite connected $C$-colored graphs such that the following holds.

\begin{enumerate}

\item
{$L$ has a polynomial time reduction to the $\mathcal{F}'$-free coloring problem.}

\item
{For every graph $G$ there are polynomially many inputs $I_1, \dots , I_k$ of $L$ computable in polynomial time (of $|G|$) such that $G$ is in the language of the $\mathcal{F}'$-free coloring problem if and only if $I_{\ell} \in L$ for every $1 \leq \ell \leq k$.}
\end{enumerate}
\end{theorem}

\begin{corollary}
The class of $\mathcal{F}'$-free coloring problems for finite $\mathcal{F}'$ with  connected underlying graphs admits no dichotomy.
\end{corollary}

\subsection{$\mathcal{F}^<$-free ordering problems have full power} \label{secfull1}

\begin{proof} (of Theorem \ref{main1})
Theorem \ref{graphs} implies that there exists a finite set of colors $C$ and a finite set of $C$-colored undirected graphs $\mathcal{F'}$ such that $L$ is computationally equivalent to the language $M$ of graphs admitting a $C$-coloring without a colored subgraph from $\mathcal{F}'$. Choose a complete graph ${\bf K} \notin M$: we may assume that such a complete graph exists, otherwise $M$ would be the class of all graphs. For a colored graph ${\bf F}'$ let ${\bf F}$ be the underlying graph without the ordering. Given a graph ${\bf G}$ define ${\bf G_*}$ to be ${\bf G}$ plus $(|C|-1)$ disjoint copies of ${\bf K}$.
Let $\mathcal{F^<}$ consist of the following ordered undirected graphs:

\begin{enumerate}

\item
{Every ordered graph containing ${\bf K}$ as a subgraph plus one 
vertex adjacent to a vertex of ${\bf K}$,} 

\item{Every ordered graph containing ${\bf K}$ as a subgraph plus one isolated vertex, where the isolated vertex is not the smallest or the largest w.r.t. the ordering,} 

\item
{The ordered graph that consists of $|C|$ disjoint copies of ${\bf K}$, where every copy of ${\bf K}$ is an interval w.r.t. the ordering,}

\item{We add for every ${\bf F}' \in \mathcal{F}'$ (possibly several) ordered graphs to ${\mathcal{F^<}}$ in the following way. For every such ordered graph the underlying graph is ${\bf F}_*$. The orderings are induced by the $C$-coloring of ${\bf F}'$ such that every copy of ${\bf K}$ is an interval w.r.t. the ordering, and the $i$th color class is the interval between the $(i-1)$th and $i$th copies of ${\bf K}$.
}
\end{enumerate}

\begin{figure}[!h]
\begin{center}
\includegraphics[width=11.0cm]{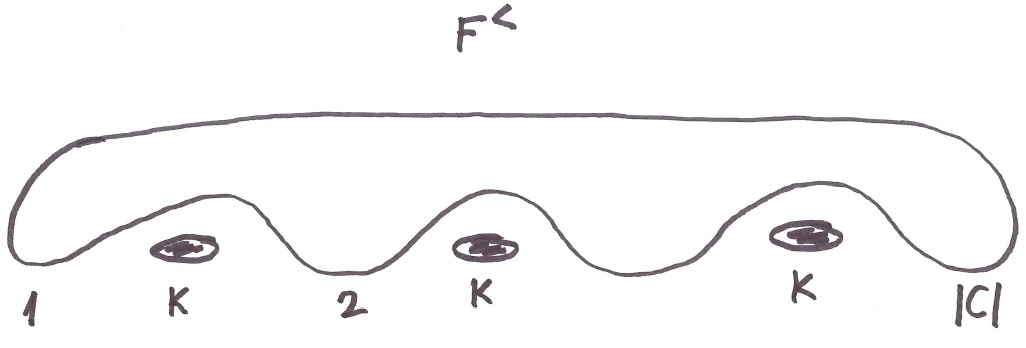}
\caption{The ordered graph ${\bf F}_*^<$}
\end{center}
\end{figure}

Note that in $(4)$ we have only one underlying graph for every ${\bf F}'$ but possibly several orderings, since the ordering inside a color class can be arbitrary. 

Let $N$ be the language of the $\mathcal{F}'$-free ordering problem. We reduce $M$ to $N$. 

\begin{claim} \label{equiorder}
${\bf G} \in M \iff {\bf G_*} \in N$
\end{claim}

\begin{proof}
One direction is straightforward: if ${\bf G} \in M$ then 
${\bf G_*} \in N$, since given a good
coloring of ${\bf G}$ we order the vertices of ${\bf G_*}$ in a way that color class $i$ is smaller than color class $j$ if $i<j$, and there is a copy of ${\bf K}$ in the ordering between consecutive color classes. This ordering witnesses that ${\bf G_*} \in N$.

Now assume that ${\bf G_*} \in N$. By the construction, the graph ${\bf G_*}$ contains exactly $(|C|-1)$ disjoint copies of ${\bf K}$, and every copy of ${\bf K}$ is an interval w.r.t. the ordering witnessing that ${\bf G_*} \in N$. Consider the coloring of ${\bf G}$, where color class $i$ consists of the vertices between the $(i-1)$th and $i$th copies of ${\bf K}$. This coloring witnesses that ${\bf G} \in M$: if there was a copy of a colored graph ${\bf F}' \in \mathcal{F}'$ in it then after adding the $(|C|-1)$ copies of ${\bf K}$ the resulting ${\bf F}_*$ with the restriction of the ordering of ${\bf G}_*$ would be in $\mathcal{F}^<$.
\end{proof}

Note that the proof used only those ordered graphs in $\mathcal{F}^<$ of type $(4)$. 

Next, we reduce $N$ to $M$.
Consider a graph ${\bf G}$. If it has a copy of
${\bf K}$ and an edge leaving it then ${\bf G} \notin N$ as witnessed by a forbidden (ordered) graph in $(1)$.
Otherwise, if it has less than $(|C|-1)$ copies of ${\bf K}$ then ${\bf G} \in N$. We may also assume that ${\bf G}$ does not contain $|C|$ disjoint copies of ${\bf K}$, else ${\bf G} \notin M$ as witnessed by  lifts added in $(2)$ and $(3)$. So it contains exactly $(|C|-1)$ disjoint copies of ${\bf K}$ as connected components of ${\bf G}$. 

Now consider the graph ${\bf H}$ we get from ${\bf G}$ by the removal of these $(|C|-1)$ copies of ${\bf K}$. Note that ${\bf G}={\bf H_*}$, hence Claim \ref{equiorder} implies that ${\bf H} \in M 
\iff {\bf G} \in N$. This completes the proof of Theorem \ref{main1}.
\end{proof}

The class of ordering problems with connected underlying graphs also has the full power of $NP$ in the following sense. This is in a remarkable contrast with Theorem \ref{main2}. 

\begin{theorem} \label{connected}
For every language $L$ in the class $NP$ there exists a finite set $\mathcal{F}^<$ of finite connected ordered graphs such that the following holds.

\begin{enumerate}

\item
{$L$ has a polynomial time reduction to the $\mathcal{F}^<$-free ordering problem.}

\item 
{For every graph $G$ there are polynomially many inputs $I_1, \dots , I_k$ of $L$ computable in polynomial time (of $|G|$) such that $G$ is in the language of the $\mathcal{F}^<$-free ordering problem if and only if $I_{\ell} \in L$ for every $1 \leq \ell \leq k$.}
\end{enumerate}
\end{theorem}

\begin{corollary}
The class of ordering problems with connected underlying graphs admits no dichotomy.
\end{corollary}

\begin{remark}
Note that if $\mathcal{F}^<$ is connected then the language of the $\mathcal{F}^<$-free ordering problem  is closed under disjoint union, hence multiple instances can be reduced to a single one. Since we can not expect this to hold for a general problem in $NP$ we need a more complicated notion of reduction like in Theorem \ref{connected}.  
\end{remark}

\section{Temporal Sparse Incomparability Lemma}\label{secTSIL}

\subsection{A combinatorial classique}
Recall that a {\it cycle} in a relational structure $\bf A$ is either a sequence of distinct elements and distinct tuples $x_0, r_1, x_1,$ $\ldots, r_t, x_t = x_0$, where each tuple $r_i$ belongs to one
of the relations $R(\bf A)$ and each element $x_i \in A$ belongs to tuples $r_i$
and $r_{i+1}$, or, in the degenerate case, $t = 1$ and $r_1$ is a relational
tuple with at least two identical coordinates. The {\it length} of the
cycle is the integer $t$ in the first case, and one in the second
case. The {\it girth} of a structure $\bf A$ is the shortest length of a cycle in $\bf A$ (if it contains a cycle, otherwise it is a
forest). The study of homomorphism properties of structures not containing short cycles is a combinatorial problem
studied intensively. The following result called {\it Sparse
Incomparability Lemma} proved to be particularly
useful in various applications.

\begin{lemma}
\label{SIL}
Let $k, \ell$ be positive integers, $\tau$ a finite relational type and $\bf B$ a finite relational structure of type $\tau$. Then there exists a finite relational structure $\underline{\bf B}$ of type $\tau$ with the following properties.

\begin{enumerate}
\item There exists a homomorphism $\underline{\bf B} \longrightarrow {\bf B}$.

\item For every structure $\bf C$ with at most $k$ elements if there exists a homomorphism $\underline{\bf B} \longrightarrow {\bf C}$ then
there exists a homomorphism ${\bf B} \longrightarrow {\bf C}$.

\item $\underline{\bf B}$ has girth at least $\ell$.
\end{enumerate}
\end{lemma}


This result is proved in \cite{NR, FV, NZ} (see also \cite{HN}) by the
probabilistic method, based on \cite{E,L}. In fact, in \cite{NR, NZ} it was proved for
graphs only but the proof is the same for finite relational
structures. The question
whether there exists a deterministic construction of the structure $\underline{\bf B}$ has been of particular interest. In the case of digraphs this has been showed in \cite{MN}, while for general relational structures a deterministic algorithm has been given in \cite{K}.

\subsection{The temporal version, outline and key ideas}
The goal of this section is to formulate a SIL for orderings of relational structures and explain the key ideas of the proof.
We consider a finite relational type $\tau$ and a {\it temporal relational structure} ${\bf T}$ of type $\tau$, by this we mean a relational structure of type $\tau$ on $\mathbb{Q}$ which has the same automorphisms as $\mathbb{Q}$. Note that an injective mapping $\iota: S \hookrightarrow \mathbb{Q}$ of a finite structure ${\bf S}$ of type $\tau$ induces an ordering on $S$: $x<_{{\bf S}}y \iff \iota(x)<\iota(y)$. An ordering $<_{{\bf S}}$ corresponds to many injective mappings, which are either all homomorphisms or none of them is a homomorphism. This equivalence allows us to switch between the language of ordered graphs and homomorphisms to ${\bf T}$.

The following result is called {\it Temporal Sparse Incomparability Lemma}. 

\begin{theorem} \label{temporalSIL}
For any integer $g$ and any relational structure ${\bf B}$ of finite type $\tau$ there is a relational structure $\underline{\bf B}$ of type $\tau$ with girth at least $g$ such that there is a homomorphism $\underline{\bf B} \rightarrow {\bf B}$, and for any temporal relational structure ${\bf T}$ of type $\tau$ we have $\underline{\bf B} \in CSP({\bf T}) \implies {\bf B} \in CSP({\bf T})$.
Moreover, $\underline{\bf B}$ can be calculated in randomized polynomial time (of $|B|$).
\end{theorem}

The proof of Theorem \ref{temporalSIL} uses the classical randomized construction for the finite SIL. However, the proof requires also other tools including the Lov\'asz Local Lemma. This might be of independent interest as an alternative approach to the finite SIL, too.

First, for the purpose of illustration, we give a randomized algorithm for the finite SIL. We use the notation $[N]=\{1, \dots ,N\}$. In this case the standard construction for finite CSPs blows up every instance ${\bf S}$ to obtain $\underline{\bf S}$ on $\underline{S}=S \times [N]$ for a suitable $N$, and if a tuple is in a relation then its first coordinates in $S$ are in the relation, too. Clearly $\underline{\bf S}$ is homomorphic to ${\bf S}$, and if $\underline{\bf S}$ is sufficiently random (under the assumption that every relational tuple is the preimage of a tuple in ${\bf S}$) then a homomorphism $\underline{\varphi}:\underline{\bf S} \to {\bf T}$ induces a homomorphism $\varphi:{\bf S} \to {\bf T}$ by setting $\varphi(s)$ to be the majority value under $\underline{\varphi}$ in the preimage of $s$ for every $s \in S$. In order to get a homomorphism $\varphi$ it is sufficient that for every relational tuple $t$ in ${\bf S}$ the majority sets in the preimage of the coordinates of the tuple also span a relational tuple. We ensure this by imposing a relation on every preimage of a relational tuple with small probability independently at random. Finally, we remove a relational tuple from every short cycle.  
The resulting $\underline{\bf S}$ has large girth, and subsets of size $N/|T|$ in the preimage of a relational tuple still span at least one relational tuple, hence the majority choice gives a homomorphism. 

Such constructions do not seem to be very useful in the case of temporal CSPs when $T$ is infinite and $\underline{\varphi}$ maps to ${\mathbb Q}$. However, we will be able to use them by modifying the definition of $\varphi$, and instead of choosing the majority value in the preimage of every element we choose a random preimage uniformly at random and use the Lov\'asz Local Lemma (LLL) to prove that this gives a homomorphism $\varphi$ (with positive probability). This argument works if the probability that the image of a relational tuple under $\varphi$ is also in the relation is $O(\Delta({\bf S})^{-1})$, where $\Delta({\bf S})$ denotes the maximum degree. 

But how to ensure such a bound on this probability for every tuple in ${\bf S}$? 
In the spirit of the hypergraph regularity lemma for semialgebraic hypergraphs by Fox, Gromov, Lafforgue, Naor, and Pach \cite{FGLNP} we prove that, given $\underline{\varphi}$ and a relational tuple $s=(s_1, \dots ,s_r) \in S^r$ such that the probability that $s$ is mapped to a relational tuple is low, there exist large subsets in the preimages of the coordinates of $s$ spanning no relational tuples. However, this can not happen for our construction. In fact, we can choose the parameters in the classical construction to push this polynomially far, such that sets of size about $\frac{N}{\Delta({\bf S})}$ (in the preimage of a relational tuple) always span a relational tuple. Thus, we get the bound needed for the LLL.

How to give a deterministic construction of $\underline{\bf S}$? For finite CSPs the first author \cite{K} has given a deterministic algorithm for $\underline{S}$. This also has the pseudorandom property that large sets in the preimage of a relational tuple span a relational tuple. However, this is not guaranteed for polynomially small subsets of $[N]$, but it is known if the subsets give a positive proportion of $[N]$. Therefore, the LLL based proof of the previous paragraph works in the case when $\Delta({\bf S})$ is bounded. Thus, it is sufficient to show that every temporal CSP is polynomially equivalent to its restriction to bounded degree relational structures, cf. Lemma \ref{bounded}. 
For more details see the full version.  

\subsection{The proof of the Temporal Sparse Incomparability Lemma} \label{secSIL}

The following proposition will be the main tool in the proof of the Temporal Sparse Incomparability Lemma, Theorem \ref{temporalSIL}. Recall that a {\it quasi-order} is a reflexive and transitive relation. Let $a(n)$ denote the number of quasi-orders of an $n$ element set, these are also known as Fubini numbers.
We write $e$ for the Euler number.

\begin{proposition} \label{homo}
Consider a finite relational type $\tau$ with maximum arity $r$, the finite relational structures ${\bf B}, \underline{\bf B}$ of type $\tau$ and a temporal relational structure ${\bf T}$ of type $\tau$. Let $\delta > 0$. 
Assume that $\underline{\bf B} \in CSP({\bf T})$, and consider a homomorphism $\underline{\bf B} \rightarrow {\bf T}$ and the induced quasi-order.  Further assume that

\begin{enumerate}

\item
{$e a(r)r(r(\Delta({\bf B})-1)+1) \delta \leq 1$, and}

\item
{there exists a mapping $\pi: \underline{\bf B} \rightarrow {\bf B}$ such that for every relational tuple $(b_1, \dots ,b_k) \in R({\bf B})$, for subintervals $S_i \subseteq \pi^{-1}(\{b_i\})$ ($1 \leq i \leq k$) w.r.t. the quasi-order on $\underline{\bf B}$ if $|S_i| > \delta |\pi^{-1}(\{b_i\})|$, then there exist $\underline{b}_i \in S_i$ ($1 \leq i \leq k$) such that $(\underline{b}_1, \underline{b}_2, \dots \underline{b}_k) \in R(\underline{\bf B})$.}

\end{enumerate}

Then  ${\bf B} \in CSP({\bf T})$.
\end{proposition}

We will not need any assumption on $\Delta(\bf B)$ in Proposition \ref{homo} in order to prove Theorem \ref{temporalSIL}, as $(1)$ can be satisfied by balancing $\delta$ for a fixed $\bf B$. Note that $\pi: \underline{\bf B} \rightarrow {\bf B}$ does not need to be a homomorphism.

The following lemma will be the key  in the proof of Proposition \ref{homo}.

\begin{lemma} \label{key}
Let $S$ be a finite set equipped with a quasi-order $\preceq_S$, $\tau$ a finite relational type, ${\bf B}$ a finite relational structure and ${\bf T}$ a temporal relational structure of type $\tau$.

Consider a mapping $\pi: S \rightarrow {\bf B}$. Let $p>0$ and let $r$ denote the maximum arity of a relational symbol in $\tau$. Assume that 
\begin{enumerate}

\item
{the inequality $ep(r(\Delta({\bf B })-1)+1) \leq 1$ holds, and}

\item
for every relational tuple $b=(b_1, \dots ,b_k) \in R(B)$ the probability, that for the elements $\underline{b}_i \in \pi^{-1}(\{b_i\})$ chosen uniformly at random the induced quasi-order of $\{b_1, \dots ,b_k\}$ defined by $b_i \preceq_{\bf B} b_j \iff s_i \preceq_S s_j$ is bad, that is, the image of the tuple $b$ is not in $R({\bf T})$, is at most $p$.

\end{enumerate}

Then ${\bf B} \in CSP({\bf T})$.
\end{lemma}

We will apply the Lov\'asz Local Lemma \cite{EL} in the proof of Lemma \ref{key}. We use the  symmetric variable version stated as Lemma~\ref{LLL}. Consider a set of mutually independent random variables. Given an event $A$ determined by these variables denote by {\it $vbl(A)$} the unique minimal set of variables that determines the event $A$:
such a set clearly exists. 

\begin{lemma} \label{LLL}
Let $\mathcal{V}$ be a finite set of mutually independent random variables in a probability space. Let $\mathcal{A}$ be a finite set of events determined by these variables. If there exist $p,d>0$ such that $ep(d+1)\leq 1$, for every $A \in \mathcal{A}$ $\mathbb{P}(A) \leq p$ and 
$|\{B: B \in \mathcal{A}, vbl(A) \cap vbl(B) \neq \emptyset\}| \leq d$, then
$\mathbb{P} \big( \bigwedge_{A \in \mathcal{A}} \overline{A} \big)>0$.
\end{lemma}

\begin{proof}(of Lemma \ref{key})
We prove that if we choose an element $f(x) \in \pi_B^{-1}(x)$ for every $x \in B$ uniformly at random then the quasi-ordering on $B$ defined by $x<_B y \iff f(x) <_S f(y)$ will with positive probability witness that ${\bf B} \in CSP({\bf T})$. 

We associate to every $x \in B$ a random variable with value $f(x)$, and to every relational tuple $t$ in ${\bf B}$ the event $A_t$ that it is badly ordered. Note that the random variables in  $vbl(A_t)$ correspond to the coordinates of $t$ (without multiplicity). 
Thus, $vbl(A_t)$ is disjoint from $vbl(A_u)$ if $t$ and $u$ do not share a coordinate, hence $vbl(A_t)$ is disjoint from all but at most $r(\Delta({\bf B})-1)$ other such sets $vbl(A_u)$. Since $ep(r(\Delta({\bf B})-1)+1) \leq 1$, Lemma \ref{LLL} shows that the probability that we avoid all the bad events, that is, the induced quasi-ordering witnesses that ${\bf B} \in CSP({\bf T})$, is positive.
\end{proof}




We will use the following basic lemma.

\begin{lemma} \label{basic}
Let $k, n$ be positive integers and $\gamma>0$. Consider the quasi-orders $\preceq_k$ on $[k]$ and $\preceq$ on $[kn]$, respectively. Assume that there are at least $\gamma k n^k$ tuples $(s_1, \dots ,s_k) \in \Pi_{i=1}^k [(i-1)k+1,ik]$ such that $s_i \preceq s_j \iff i \preceq_k j$ for every $i,j$. Then there are sets $S_i \subseteq [(i-1)k+1,ik]$ of size at least $\gamma n$ such that $s_i \preceq s_j \iff i \preceq_k j$ for every $i,j \in [k]$ and $s_i \in S_i, s_j \in S_j$.
\end{lemma}

\begin{proof}
We may assume that $i<j \implies i \preceq_k j$ for $i,j \in [k]$.

We prove by induction on $k$, the case $k=1$ is trivial. Assume that the statement holds for integers less than $k$, and consider a set $R$ of at least $\gamma k n^k$ such $k$-tuples. 
If $k \npreceq_{k} (k-1)$ then let $S_k$ be the set of the last $\gamma n$ elements of $[(k-1)n+1,kn]$ w.r.t. $\preceq$. There are at least $\gamma (k-1) n^k$ $k$-tuples in $R$ whose last coordinate is not in $S_k$, hence it should be smaller than or equal to any element of $S_k$ w.r.t. $\preceq$. Consider the set $R_{k-1}$ of $(k-1)$-tuples obtained by the removal of the last coordinate: $R_{k-1}$ contains at least $\gamma (k-1) n^{k-1}$ many $(k-1)$-tuples, so by induction there are sets $S_1, \dots ,S_{k-1}$ satisfying to the conditions. Since any element in $S_{k-1}$ is strictly smaller than any element of $S_k$, the sets $S_1, \dots ,S_k$ will be as required.

Now assume that there exists $\ell \lneq k$ such that $k \preceq_k \ell \npreceq_k \ell-1$. Observe that the number of tuples $s \in R$, such that the $\preceq$ equivalence class of $s_{\ell}, \dots ,s_k$ intersects at least one of the intervals $[(\ell-1)n+1,\ell n], \dots ,[(k-1)n+1,kn]$ in at most $\gamma n$ elements, is at most $(k-\ell+1)\gamma n^k$, since if we fix all but one of these elements then there are at most $\gamma n$ possibilities to choose the remaining one. Let $R' \subseteq R$ denote the set of the other $k$-tuples, we have $|R'| \geq (\ell-1)\gamma n^k$. Let us choose the largest equivalence class $E$ w.r.t. $\preceq$ with size at least $\gamma n$ in each of these $k-\ell +1$ intervals, and set $S_m = E \cap [(m-1)n+1,mn]$ for $\ell \leq m \leq k$. 

Consider the set $R_{\ell-1}$ of $(\ell-1)$-tuples obtained by the removal of the last $(k-\ell+1)$ coordinates from a tuple in $R'$. Since $|R'|\geq (\ell-1)\gamma n^k$ we have $|R_{\ell-1}|\geq (\ell-1)\gamma n^{k-\ell-1}$, so by induction there are sets $S_1, \dots ,S_{\ell-1}$ of size at least $\gamma n$ that for any choice of elements from these sets the induced quasi-order will be the restriction of $\preceq_k$ to $[\ell-1]$. And by the choice of $S_{\ell}, \dots ,S_k$ the set $S_{\ell-1}$ will be strictly smaller than these w.r.t. $\preceq$, hence the sets $S_1, \dots ,S_k$ will be as required.
\end{proof}

The semialgebraic hypergraph regularity lemma of Fox, Gromov, Lafforgue, Naor, and Pach \cite{FGLNP} 
provides an alternative tool to prove a variant of this lemma, see also Tidor and Yu \cite{TY}.

\begin{proof}(of Proposition \ref{homo})
Assume that $\underline{\bf B} \in CSP({\bf T})$, we will consider the witness homomorphism $\underline{B} \rightarrow 
{\bf T}$ and the induced quasi-order on $\underline{B}$.

\begin{claim} \label{sparse}
Consider the $k$-ary relation $R$ and
the tuple $(b_1, \dots ,b_k) \in R({\bf B})$. The probability that for $\underline{b}_1 \in \pi^{-1}(\{ b_1 \}), \dots \underline{b}_k \in \pi^{-1}(\{ b_k \})$ chosen uniformly at random the tuple $\underline{b}$ is badly quasi-ordered is at most $a(k) k \delta$.
\end{claim}

\begin{proof}
We prove by contradiction. Suppose that this probability is greater than $a(k) k \delta$. Then it is greater than $k \delta$ for one of the bad quasi-orders $\preceq_k$ of $[k]$. We may assume that the preimage of every element under $\pi$ has the same size: we obtain this after replacing $\underline{\bf B}$ with its blow-up, what does not change the probabilities in question. Now Lemma \ref{basic} provides subsets $S_i \subseteq \pi^{-1}(\{ b_i \})$ with size $|S_i| > \delta |\pi^{-1}(\{ b_i \})|$ such that the restriction of the quasi-order is $\preceq_k$ on any $k$-tuple chosen from these sets. We may assume that the sets $S_1, \dots ,S_k$ are intervals w.r.t. to the quasi-order on $\underline{\bf B}$. However, by assumption $(2)$ of Proposition \ref{homo} these sets span a tuple in relation $R$, contradicting the fact that the quasi-order witnesses that $\underline{\bf B} \in CSP({\bf T})$. 
\end{proof}

Consider a $k$-ary relational symbol 
$R \in \tau$. Lemma \ref{sparse} provides for every $(x_1, \dots ,x_k) \in R({\bf B})$ the estimate $|\{(y_1, \dots ,y_k): (y_1, \dots ,y_k) \notin R(\underline{\bf B}), \forall i \text{ } \pi(y_i)=x_i \}| \leq \delta k a(k) \Pi_{i=1}^k |\pi^{-1}(\{ x_i \})|$. In other words, the probability that a random preimage of this relational tuple is badly quasi-ordered in $\underline{\bf B}$ is at most $\delta k a(k)$.
Every relational tuple has arity at most $r$. Thus, we can apply Lemma \ref{key} to $p=r a(r) \delta$
since $ep(r(\Delta({\bf B})-1)+1) \leq 1$.
\end{proof}

\begin{proof}(of Theorem \ref{temporalSIL})
We will find a structure $\underline{\bf B}$ with girth greater than $g$ that satisfies both assumptions of Proposition \ref{homo} for a $\delta>0$.  

We will use the standard randomized construction for the SIL to get a relational structure $\underline{\bf B}$ with girth at least $g$. We refine \cite{FV, NR} who adapted \cite{E}. This will give a randomized polynomial time construction of $\underline{\bf B}$. Set $\delta=e^{-1}r^{-r-2}|\tau|^{-1}|B|^{-r}$, so the assumption of $(1)$ of Proposition \ref{homo} will be satisfied, since $\Delta({\bf B}) \leq |\tau| |B|^r$. 

First, let the base set be $\underline{B} = B \times \{1, \dots ,n\}$ for $n$ large enough (but a polynomial of $|B|$) chosen later. Consider the projection $\pi:\underline{B} \to B$. And let us choose $p_1, \dots ,p_r>1$ also later. Let $\underline{\bf B}_0$ be the following random structure with base set $\underline{B}$. Given a $k$-ary relational symbol $R$, a relational tuple $b \in R({\bf B})$ and $\underline{b} \in \underline{B}^k$, where $\pi(\underline{b}_i)=b_i$, add $\underline{b}$ to $R(\underline{\bf B}_0)$ with probability $p_k$, independently for every relational symbol $R$ and pair of tuples $b, \underline{b}$.

Finally, remove a relational tuple of $\underline{\bf B}_0$ in every cycle with length at most $g$  in $\underline{\bf B}_0$ in order to get the structure $\underline{\bf B}$ with girth at least $g$. We need to check that assumption $(2)$ of Proposition \ref{homo} holds (with high probability).

Put $p_j=n^{1-j+1/g}$. The number of $k$-cycles in ${\bf B}$ with tuples from $R_1, \dots ,R_k$ with arities $r_1, \dots ,r_k$, respectively,  is at most $|B|^k |B|^{\sum_{i=1}^k(r_i-2)}$, 
so the expected number of relational tuples in such $k$-cycles removed from $\underline{\bf B}_0$ is at most $\Pi_{i=1}^k p_i \cdot n^{\sum_{i=1}^k (r_i-1)} \cdot |B|^{\sum_{i=1}^k (r_i-1)}=n^{k/g} \cdot \Pi_{i=1}^k |B|^{\sum_{i=1}^k (r_i-1)}$. Therefore, the expected number of all tuples removed is $O(|B|^{g(r-1)} n)$, where the constant hidden in $O(*)$ depends on $\tau$ and $g$ only.

Given a $k$-ary relation $R({\bf B})$, a relational tuple $b \in R({\bf B})$ and for $i=1, \dots ,k$ subsets 
$S_i \subseteq \pi^{-1}(b_i)$, the expected number of tuples in $R(S_1, \dots ,S_k)$ is $p_k \Pi_{i=1}^k |S_i|$.
If $|S_i| \geq \delta n$ for every $i$ then this is at least 
$p_k \delta^k n^k = p_k\big(er^{r+2}|\tau||B|^{r}\big)^{-k}n^k = \big(er^{r+2}|\tau||B|^{r}\big)^{-k}n^{1+1/g}$.

Choose $n=|B|^{3g^2r}$, so for any such $(S_1, \dots ,S_k)$ the expected value of $R(S_1, \dots ,S_r)$ is at least $n^{1+\frac{1}{2g}}$ if $|B|$ is large enough. We can choose the $k$ sets in $O(|B|^r 2^{rn})$ ways. The probability, that the number of tuples spanned by them, is less than half of the expected value is less than an exponentially small function of $n^{1+\frac{1}{2g}}$ by the Chernoff bound. Thus, with high probability, for every choice of $(S_1, \dots ,S_k)$ they span at least half of the number of expected tuples with high probability. 

The number of tuples removed is with high probability much smaller than this by the Markov inequality, since its expected value is already much smaller if $|B|$ is large enough. Hence assumption $(2)$ of Proposition \ref{homo} holds with high probability for $\underline{\bf B}$. This completes the proof of the Theorem \ref{temporalSIL}. 
\end{proof}

\section*{Acknowledgments.}
We would like to thank the anonymous referees for their remarks and suggestions.


\end{document}